\documentclass[10pt,conference,compsocconf,letterpaper]{IEEEtran}
\usepackage{verbatim}
\usepackage{graphicx}
\usepackage{times}
\usepackage{amsmath,amsfonts,amssymb,latexsym}
\usepackage{color}
\usepackage{subfigure}
\usepackage{macros}
\usepackage[ruled]{algorithm}
\usepackage{algpseudocode}
\usepackage{ioa_code}
\usepackage[hyphens]{url}

\usepackage{listings, color}          
\definecolor{dkgreen}{rgb}{0,0.6,0}
\definecolor{gray}{rgb}{0.5,0.5,0.5}
\definecolor{mauve}{rgb}{0.58,0,0.82}
\definecolor{blue}{rgb}{0.13, 0.58, 0.99}
\definecolor{black}{rgb}{0,0,0}
\definecolor{darkgray}{rgb}{0.3, 0.3, 0.3}
\definecolor{red}{rgb}{0.8, 0.2, 0.2}

\lstdefinelanguage{JavaScript}{%
  keywords={typeof, new, true, false, catch, function, modifier, contract, return, null, catch, switch, var, if, in, while, do, else, case, break, throw},
  keywordstyle=\color{blue}\bfseries,
  ndkeywords={class, export, boolean, implements, import, this, uint32, uint256, mapping, address, public},
  ndkeywordstyle=\color{darkgray}\bfseries,
  identifierstyle=\color{black},
  sensitive=false,
  comment=[l]{//},
  morecomment=[s]{/*}{*/},
  commentstyle=\color{dkgreen}\ttfamily,
  stringstyle=\color{red}\ttfamily,
  morestring=[b]',
  morestring=[b]"
}

\usepackage[justification=centering]{caption} 

\usepackage{color}
\usepackage[usenames,dvipsnames,svgnames,table]{xcolor}
\usepackage[colorlinks=true,
            linkcolor=gray,
            urlcolor=blue,
            citecolor=red]{hyperref}

\usepackage{amsthm}
\newtheorem{theorem}{Theorem}

\newtheorem{lemma}[theorem]{Lemma}
\newtheorem{definition}[theorem]{Definition}
\newtheorem{fact}[theorem]{Fact}

\newcommand{\vincent}[1]{}

\newcommand{\remove}[1]{}

\begin{document}
\title{
The Balance Attack Against Proof-Of-Work Blockchains: The R3 Testbed as an Example}

\author{\IEEEauthorblockN{Christopher Natoli}
\IEEEauthorblockA{\small Data61-CSIRO\\
University of Sydney\\
\small{christopher.natoli@sydney.edu.au} }
\and
\IEEEauthorblockN{Vincent Gramoli}
\IEEEauthorblockA{\small Data61-CSIRO\\
University of Sydney\\
\small{vincent.gramoli@sydney.edu.au}}}
%

\date{}

\maketitle

\thispagestyle{plain}
\pagestyle{plain}

\begin{abstract}
%

In this paper, we identify a new form of attack, called the Balance attack, against proof-of-work blockchain systems. The novelty of this attack consists of delaying network communications between multiple subgroups of nodes with balanced mining power. Our theoretical analysis captures the precise tradeoff between the network delay and the mining power of the attacker needed to double spend in Ethereum with high probability. 

We quantify our probabilistic analysis with statistics taken from the R3 consortium, and show that a single machine needs 20 minutes to attack the consortium.
Finally, we run an Ethereum private chain in a distributed system with similar settings as R3 to demonstrate the feasibility of the approach, and discuss the application of the Balance attack to Bitcoin. Our results clearly confirm that main proof-of-work blockchain protocols can be badly suited for consortium blockchains. \\

\end{abstract}

\section{Introduction}

Blockchain systems are distributed implementations of a chain of blocks.
Each node can issue a cryptographically signed transaction to transfer digital assets
to another node or can create a new block of transactions, by solving a crypto-puzzle,  
and append this block to its current 
view of the chain. Due to the distributed nature of this task, multiple nodes may append
distinct blocks at the same index of the chain before learning about the presence of other blocks, 
hence leading to a forked chain or a \emph{tree}.
For nodes to eventually agree on a unique state of the system, nodes apply a common strategy that 
selects a unique branch of blocks in this tree.

Bitcoin~\cite{Nak08}, one of the most popular blockchain systems, selects the longest branch. This strategy has however shown its limitation as it simply \emph{wastes 
all blocks} not 
present in this branch~\cite{DW13,PSS16,SZ15,GKK16,NKMS16}. 
If an attacker can solve crypto-puzzles fast enough to grow a local branch of the blockchain faster 
than the rest of the system, then it will eventually impose its own branch to all participants.
In particular, by delaying the propagation of blocks in the system, one can increase the amount of 
wasted blocks and proportionally slow down the growth of the longest branch of the system. 
This delay presents a serious risk to the integrity of the blockchain, 
as the attacker does not even need a large fraction of the computational power to exceed the length of 
the chain, allowing her to \emph{double spend} in new transactions the coins that she already spent in earlier transactions~\cite{Ros12}.
 
Ethereum~\cite{Woo15} proposes another selection strategy that copes with this problem. 
Each node uses an algorithm, called {\sc Ghost}, that starts from the first block, also called the \emph{genesis block}, and 
iteratively selects the root of the heaviest subtree to construct the common branch.
Even if nodes create many blocks at the same index of the blockchain, their computational power is
not wasted but counted in the selection strategy~\cite{SZ15}. 
In particular, the number of these ``sibling'' blocks 
increase the chance that their common ancestor block be selected in favor of another 
candidate block mined by the attacker.
Although it clearly alleviates the Bitcoin limitation discussed  
above~\cite{DW13,PSS16,GKK16,NKMS16} it remains unclear how long an attacker with a low mining power should delay messages to discard previous transactions in Ethereum.

In this paper, we answer this question by demonstrating theoretically and experimentally that an attacker can compensate a low mining power by delaying selected 
messages in Ethereum.
To this end, we propose a simple attack, called the \emph{Balance Attack}: an attacker transiently 
disrupts communications between subgroups of similar mining power. 
During this time, the attacker issues transactions in one subgroup, say the 
\emph{transaction subgroup}, 
and mines blocks in another subgroup, say the \emph{block subgroup}, up to the point where the tree 
of the block subgroup outweighs, with 
high probability, the tree of the transaction subgroup. 
The novelty of the Balance attack is to leverage the {\sc Ghost} protocol that accounts for sibling or 
\emph{uncle} blocks to select a chain of blocks. 
This strategy allows the attacker to mine a branch possibility in isolation of the rest of the 
network before merging its branch to one of the competing blockchain to influence the branch 
selection process.

We experimented a distributed system running Ethereum in similar settings as R3, a consortium of 
more than 70 world-wide financial institutions.
In January, R3 consisted of eleven banks and successfully collaborated in deploying an Ethereum  private chain to perform transactions.
\footnote{\url{http://www.ibtimes.co.uk/r3-connects-11-banks-distributed-ledger-using-ethereum-microsoft-azure-1539044}.}
Since then, R3 has grown and kept experimenting Ethereum\footnote{\url{http://www.coindesk.com/r3-ethereum-report-banks/}.} and other technologies while the concept of \emph{consortium private chain} gained traction for its ability to offer a blockchain system among multiple companies in a private and controlled environment. 
R3 has just released their own Corda framework.
%
As opposed to a fully private chain scenario, the consortium private chain involves 
different institutions possibly competing among each other. As they can be located 
at different places around the world, they typically use Internet to communicate.
%
%
We illustrate the Balance attack 
in the R3 testbed setting as of June 2016, by
deploying our own private chain on 15 mining virtual machines in Emulab and
configuring the network with ns-2. 
%
%

While disrupting the communication between subgroups of a blockchain system may look difficult, 
there have been some attacks successfully delaying messages of Bitcoin in the past.
In 2014, a BGP hijacker exploited access to an ISP to steal \$83000 worth of bitcoins by positioning itself between Bitcoin pools and their miners~\cite{LS14}.
Some known attacks against Bitcoin involved partitioning the communication graph at the network level~\cite{AZV16} and at the application level~\cite{HKZG15}.
%
At the network level, a study indicated the simplicity for autonomous systems to intercept a large amount of 
bitcoins and evaluated the impact of these network attacks on the Bitcoin protocol~\cite{AZV16}.
%
%
At the application level, some work showed that an attacker controlling 32 IP addresses 
can ``eclipse'' a Bitcoin node with 85\% probability~\cite{HKZG15}. 
%
More generally, man-in-middle attacks can lead to similar results by relaying the traffic between 
two nodes through the attacker.

One can exploit the Balance attack to violate the persistence of the main branch, hence
rewriting previously committed transactions, and allowing the attacker to double spend.
As opposed to previous attacks against Bitcoin where the attacker has to expand the longest chain 
faster than correct miners to obtain this result~\cite{Ros12}, the novelty of our attack lies in the contribution of the attacker to one of the correct miner chain in order to outweigh another correct miner chain of Ethereum.
We generalize our contribution to
proof-of-work algorithms by 
proposing a simple model for proof-of-work blockchains and specifying Nakamoto's and {\sc Ghost} 
consensus algorithmic differences. We also discuss how to adapt the Balance attack to 
violate the persistence of Bitcoin main branch. This adaptation requires to mine at the top of 
one of the correct chains rather than solo-mining a subchain but can lead to similar 
consequences.
More precisely, we make the four following contributions:
\begin{enumerate}
\item We show that the {\sc Ghost} consensus protocol can be vulnerable to a double spending attack without coalition and with high probability if a single attacker can delay communication between multiple communication subgraphs while owning 5\% of the total mining power of the system. \vincent{Simple distributed model, difference between bitcoin and ethereum in commit and selection.}
\item We illustrate the problem in the context of the R3 consortium blockchain, as of June 2016, where we observed 50 nodes among which 15 were mining. We show that one of the nodes can execute a Balance attack by disrupting some communication channels less than 4 minutes. 
\item We demonstrate a tradeoff between the mining power needed and the time selected communication channels have to be delayed to attack Ethereum. This suggests that combining network-level with application-level attacks increases the Ethereum vulnerability.
\item We generalize this result to Nakamoto's protocol and propose an adaptation of the Balance 
attack to double spend 
in Bitcoin if the attacker can contribute, even with a small power, by mining on top of some of the correct chains.
\end{enumerate}



Section~\ref{sec:prel} defines the problem.
In Section~\ref{sec:pa}, we present the algorithm to run the attack.
In Section~\ref{sec:proof}, we show how the analysis affects {\sc Ghost}.
In Section~\ref{sec:example}, we simulate a Balance attack in the context of the R3 consortium network.
In Section~\ref{sec:expe}, we present our experiments run in an Ethereum private chain.
In Section~\ref{sec:disc}, we discuss the implications of the attack in existing blockchain systems.
Section~\ref{sec:rw} presents the related work.
And Section~\ref{sec:conclusion} concludes.
Appendix~\ref{app:proof} includes the missing proofs for the general case.


\section{Preliminaries}\label{sec:prel}

\vincent{main chain or branch?}
In this section we model a simple distributed system as a communication graph that 
implements a blockchain abstraction as a directed acyclic graph. 
We propose a high-level pseudocode representation of proof-of-work blockchain protocols in 
this model that allows us to illustrate an important difference between Bitcoin and 
Ethereum in the selection of a main branch with a persistent prefix.

\subsection{A simple distributed model for blockchains}

We consider a communication graph $G = \tup{V,E}$ with nodes $V$
connected to each other through fixed communication links $E$. 
Nodes are part of a blockchain system $S \in \{\lit{bitcoin}, \lit{ethereum}\}$ and can act as clients by issuing transactions to the system and/or servers by \emph{mining}, the 
action of trying to combine transactions into a block.
For the sake of simplicity, we consider that each node possesses a single account and that a 
\emph{transaction} issued by node $p_i$ is a transfer of digital assets or \emph{coins} 
from the account of the source node $p_i$ to the account of 
a destination node $p_j \neq p_i$.
Each transaction is uniquely identified and broadcast to all nodes in a best-effort manner.
We assume that a node re-issuing the same 
transfer multiple times creates as many distinct transactions.

Nodes that mine are called \emph{miners}. We refer to the computational power of a miner as its \emph{mining power} and we denote the total mining power $t$ as the sum of the mining powers of all miners in $V$.
Each miner tries to group a set $T$ of transactions it heard about into a block $b \supseteq T$ as long as transactions of $T$ do not conflict and that the account balances remain non-negative.
For the sake of simplicity in the presentation, the graph $G$ is static meaning that no nodes can join and leave the system, however, nodes may fail as described in Section~\ref{ssec:faults}.
\vincent{We consider that the network is asynchronous in that there is no bound on the delay it takes for the 
message to be delivered.}

\subsubsection{Miner must solve a crypto-puzzle to create a new block}
Miners provably solve a hashcash crypto-puzzle~\cite{Bla02} before creating a new block. 
Given a global threshold and the block of largest index the miner knows, the miner repeatedly selects a nonce and applies a pseudo-random function to this block and the selected nonce until it obtains a result lower than the threshold.
Upon success the miner creates a block that contains the successful nounce as a proof-of-work as 
well as the hash of the previous block, hence fixing the index of the block, and broadcasts the block.
As there is no known strategy to solve the crypto-puzzle, the miners simply keep testing whether randomly chosen numbers solve the crypto-puzzle. 
The mining power is thus expressed in the number of hashes the miner can test per second, or $H/s$ for short.
The \emph{difficulty} of this crypto-puzzle, defined by the threshold, limits the rate at which new blocks can be generated by the network.
In the remainder, we refer to $d$ as the difficulty of this crypto-puzzle.
%

\subsubsection{The failure model}\label{ssec:faults}
We assume the presence of an \emph{adversary} (or attacker)
that can control nodes that together own a relatively small fraction $\rho < 0.5$ of the total
mining power of the system.
The nodes controlled by the adversary are called \emph{malicious} and may not follow the protocol specification, however, they cannot impersonate other nodes while issuing transactions.\footnote{This is typically ensured through public key crypto-systems.}
A node that is not malicious is \emph{correct}.
We also assume that the adversary can transiently disrupt communications on a selected subset of edges $E_0$ of the communication graph $G$.
\vincent{TODO: Replace attacker by adversary everywhere.}


\begin{figure}[t]
\begin{center}
\subfigure[view $\ell_1$\label{sfig:l1}]{\includegraphics[clip=true, viewport=150 0 280 200, scale=0.35]{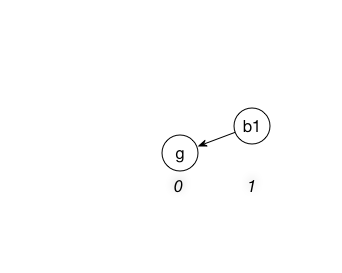}}
\subfigure[view $\ell_2$\label{sfig:l2}]{\includegraphics[clip=true, viewport=130 40 260 200, scale=0.35]{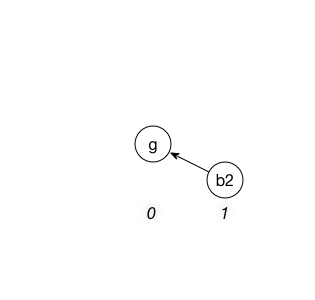}}
\subfigure[view $\ell_3$\label{sfig:l3}]{\includegraphics[clip=true, viewport=130 0 330 200, scale=0.35]{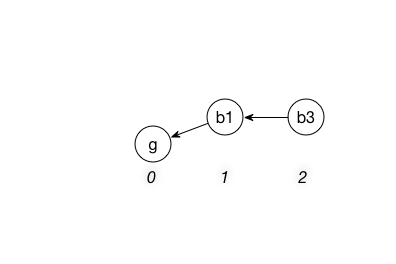}}
\subfigure[global state $\ell_0=\ell_1 \cup \ell_2 \cup \ell_3$\label{sfig:union}]{\includegraphics[clip=true, viewport=130 40 330 205, scale=0.35]{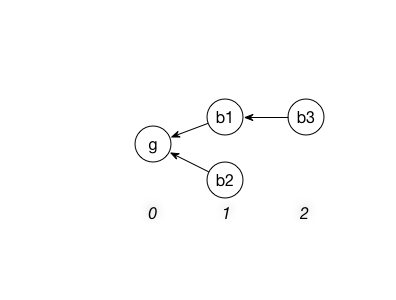}}
\caption{The global state $\ell_0$ of a blockchain results from the union of the distributed local views $\ell_1$, $\ell_2$ and $\ell_3$ of the blockchain\label{fig:chains}}
\end{center}
\end{figure}

\subsubsection{The blockchain abstraction}
Let the \emph{blockchain} be a directed acyclic graph (DAG) $\ell = \tup{B, P}$ such that 
blocks of $B$ point to each other with pointers $P$ (pointers are recorded in a block as a hash of the previous block) and a special block $g \in B$, called the \emph{genesis block}, does not point to any  block. 
%

\begin{algorithm}[!ht]
  \caption{Blockchain construction at node $p_i$}\label{alg:pow}
  \begin{algorithmic}[1]
    {\footnotesize
	  \State $\ell_i = \tup{B_i, P_i}$, the local blockchain at node $p_i$ is a directed acyclic 
	  \State \T graph of blocks $B_i$ and pointers $P_i$
	 
	\Statex
	
	\Part{$\act{receive-blocks}(\tup{B_j, P_j})_i$}{ \label{line:pow-receive-block-starts} \Comment{upon reception of blocks}
	        \State $B_i \gets B_i \cup B_j$ \Comment{update vertices of blockchain}
		\State $P_i \gets P_i \cup P_j$ \Comment{update edges of blockchain}\label{line:pow-receive-end}
	}\EndPart 
	
%
	
    }
    \algstore{pow}
    \algstore{pow2}
  \end{algorithmic}
\end{algorithm}

Algorithm~\ref{alg:pow} describes the progressive construction of the blockchain at a 
particular node $p_i$ upon reception of blocks from other nodes by simply aggregating 
the newly received blocks to the known blocks  (lines~\ref{line:pow-receive-block-starts}--\ref{line:pow-receive-end}).
As every added block contains a hash to a previous block that eventually leads back to the genesis block, each block is associated with a fixed index. By convention we consider the genesis block at index $0$, and the blocks at $k$ hops away from the genesis block as the blocks at index $k$. 
As an example, consider the simple blockchain $\ell_1 = \tup{B_1, P_1}$ depicted in Figure~\ref{sfig:l1} where $B_1 = \{g, b_1\}$ and $P_1 = \{\tup{b_1, g}\}$. The genesis block $g$ 
has index 0 and the block $b_1$ has index 1. 


\subsubsection{Forks as disagreements on the blocks at a given index}

As depicted by views $\ell_1$, $\ell_2$ and $\ell_3$ in Figures~\ref{sfig:l1},~\ref{sfig:l2} and~\ref{sfig:l3}, respectively, nodes may have a different views 
of the current state of the blockchain.
In particular, it is possible for two miners $p_1$ and $p_2$ to mine almost simultaneously two different blocks, say $b_1$ and $b_2$.
If neither block $b_1$ nor $b_2$ was propagated early enough to nodes $p_2$ and $p_1$, respectively, then 
both blocks would point to the same previous block $g$ as depicted in Figures~\ref{sfig:l1} 
and~\ref{sfig:l2}.
Because network delays are not predictable, a third node $p_3$ may receive the block $b_1$ 
and mine a new block without hearing about $b_2$. The three nodes $p_1$, $p_2$ and $p_3$ thus end up having 
three different local views of the same blockchain, denoted $\ell_1 = \tup{B_1, P_1}$, $\ell_2= \tup{B_2, P_2}$ and $\ell_3= \tup{B_3, P_3}$. 

We refer to the \emph{global blockchain} as 
the directed acyclic graph $\ell_0 = \tup{B_0, P_0}$ representing the union of these local blockchain views, denoted by $\ell_1 \cup \ell_2 \cup \ell_3$ for short, as depicted in Figure~\ref{fig:chains}, and more formally defined as follows:
\begin{equation}
\left\{\begin{array}{ll}
B_0 &= \cup_{\forall i} B_i,\notag \\
P_0 &= \cup_{\forall i} P_i.
\end{array}\right.
\end{equation}
The point where distinct blocks of the global blockchain DAG have the same predecessor block
is called a \emph{fork}. 
As an example Figure~\ref{sfig:union} depicts 
a fork
with two branches pointing to the same block: $g$ in this example. 

In the remainder of this paper, we refer to the DAG as a \emph{tree} rooted in $g$ 
with upward pointers, where children blocks point to their parent block.

\subsubsection{Main branch in Bitcoin and Ethereum}\label{ssec:consensus}
To resolve the forks and define a deterministic state agreed upon by all nodes, a blockchain system must select a \emph{main branch}, as a unique sequence of blocks, based on the tree. 
Building upon the generic construction (Alg.~\ref{alg:pow}),
we present 
two 
selections: Nakamoto's consensus protocol (Alg.~\ref{alg:nakamoto}) present in Bitcoin~\cite{Nak08} and the 
{\sc Ghost} consensus protocol (Alg.~\ref{alg:ghost}) present in Ethereum~\cite{Woo15}. 


%

\paragraph{Nakamoto's consensus algorithm}
The difficulty of the crypto-puzzles used in Bitcoin produces a block every 10 minutes in expectation.
The advantage of this long period, is that it is relatively rare for the blockchain to fork because blocks are rarely mined during the time others are propagated to the rest of the nodes. 

\begin{algorithm}[!ht]
  \caption{Nakamoto's consensus protocol at node $p_i$}\label{alg:nakamoto}
  \begin{algorithmic}[1]
  \algrestore{pow}
    {\footnotesize

%
%
%

	\State $m = 5$, the number of blocks to be appended after the block containing\label{line:m-btc}
	\State \T $\ms{tx}$, for $\ms{tx}$ to be committed in Bitcoin 
	\Statex
	
	\Part{$\act{get-main-branch}()_i$}{    	\Comment{select the longest branch} \label{line:btc-pruning-start}
		\State $b \gets \lit{genesis-block}(B_i)$ \Comment{start from the blockchain root}
		\While{$b.\ms{children} \neq \emptyset$} \Comment{prune shortest branches}
			\State $\ms{block} \gets \lit{argmax}_{c\in b.\ms{children}}\{\lit{depth}(c)\}$  \Comment{root of deepest subtree}\label{line:deepest}
			\State $B \gets B \cup \{\ms{block}\}$ \Comment{update vertices of main branch}
			\State $P \gets P \cup \{\tup{\ms{block},\ms{b}}\}$ \Comment{update edges of main branch}
			\State $b \gets \ms{block}$ \Comment{move to next block}
		\EndWhile
		\State {\bf return} $\tup{B,P}$ \Comment{returning the Bitcoin main branch} \label{line:btc-pruning-end}
	}\EndPart
	
	\Statex
	
	\Part{$\act{depth}(b)_i$}{ 				\Comment{depth of tree rooted in $b$}
		\If{$b.\ms{children} = \emptyset$} {\bf return} 1 \Comment{stop at leaves}
		\Else{}
			 {\bf return} 1 + $\max_{c \in b.\ms{children}} \lit{depth}(c)$ \Comment{recurse at children}
		\EndIf	\label{line:btc-depth-end}
	}\EndPart
	
%
	
    }
  \end{algorithmic}
\end{algorithm}

Algorithm~\ref{alg:nakamoto} depicts the Bitcoin-specific pseudocode that includes Nakamoto's consensus protocol to decide on a particular block at index $i$ (lines~\ref{line:btc-pruning-start}--\ref{line:btc-depth-end}) and the choice of parameter $m$ (line~\ref{line:m-btc}) explained later in Section~\ref{ssec:termination}.
When a fork occurs, Nakamoto's protocol resolves it by selecting the deepest branch as the main branch (lines~\ref{line:btc-pruning-start}--\ref{line:btc-pruning-end}) by iteratively selecting the root of the deepest subtree (line~\ref{line:deepest}).  When process $p_i$ is done with this pruning, the resulting branch becomes the main branch $\tup{B_i, P_i}$ as observed by the local process $p_i$.
Note that 
the pseudocode for checking whether a block is decided and a transaction committed based on this parameter $m$ is common to Bitcoin and Ethereum, it is thus deferred to Alg.~\ref{alg:commit}.
%

\vincent{
Bitcoin guarantees that a transaction is immutable as soon as it is committed, provided that 
blocks are propagated in a limited amount of time.
It also guarantees that if the recipient of the fund transferred in a transaction observes that the transaction is included in a block followed by $k$ block, then the transaction is committed.
}

\subsubsection{The {\sc Ghost} consensus algorithm}
As opposed to the Bitcoin protocol, Ethereum generates one block every 12--15 seconds. While it improves the throughput (transactions per second) it also favors transient forks as miners are more likely to propose new blocks without having heard about the latest mined blocks yet.
To avoid wasting large mining efforts while resolving forks, Ethereum uses the {\sc Ghost} (Greedy Heaviest Observed Subtree) consensus algorithm that accounts for the, so called \emph{uncles}, blocks of discarded branches.
In contrast with Nakamoto's protocol, the {\sc Ghost} protocol iteratively selects, as the successor block, the root of the subtree that contains the largest number of nodes (cf. Algorithm~\ref{alg:ghost}).
%
%
%
%
%

\begin{figure}[ht!]
\begin{center}
\includegraphics[scale=0.5, clip=true, viewport=120 50 600 300]{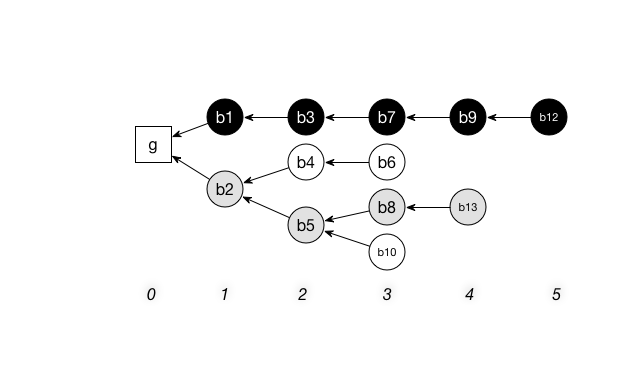}
\caption{Nakamoto's consensus protocol at the heart of Bitcoin selects the main branch as the deepest branch (in black) whereas the {\sc Ghost} consensus protocol at the heart of Ethereum follows the heaviest subtree (in grey)\label{fig:reorg}}
\end{center}
\end{figure}

The main difference between Nakamoto's consensus protocol and {\sc Ghost} is depicted in Figure~\ref{fig:reorg}, where the black blocks represent the main branch selected by Nakamoto's consensus protocol and the grey blocks represent the main branch selected by {\sc Ghost}.

\begin{algorithm}[!ht]
  \caption{The {\sc Ghost} consensus protocol at node $p_i$}\label{alg:ghost}
  \begin{algorithmic}[1]
  \algrestore{pow2}
    {\footnotesize

	\State $m = 11$, the number of blocks to be appended after the block containing\label{line:m-eth}
	\State \T $\ms{tx}$, for $\ms{tx}$ to be committed in Ethereum (since Homestead v1.3.5)
	\Statex

%
%
%
	
	\Part{$\act{get-main-branch}()_i$}{    	\Comment{select the branch with the most nodes}\label{line:eth-get-branch-start}
		\State $b \gets \lit{genesis-block}(B_i)$ \Comment{start from the blockchain root}
		\While{$b.\ms{children} \neq \emptyset$} \Comment{prune lightest branches}
			\State $\ms{block} \gets \lit{argmax}_{c\in b.\ms{children}}\{\lit{num-desc}(c)\}$ \Comment{root of heaviest tree}
			\State $B \gets B \cup \{\ms{block}\}$ \Comment{update vertices of main branch}
			\State $P \gets P \cup \{\tup{\ms{block},\ms{b}}\}$ \Comment{update edges of main branch}
			\State $b \gets \ms{block}$ \Comment{move to next block}
		\EndWhile
		\State {\bf return} $\tup{B,P}$. \Comment{returning the Ethereum main branch}\label{line:eth-get-branch-end}
	}\EndPart
	
	\Statex
	
	\Part{$\act{num-desc}(b)_i$}{  \Comment{number of nodes in tree rooted in $b$}
		\If{$b.\ms{children} = \emptyset$} {\bf return} 1 \Comment{stop at leaves}\label{line:num-desc-start}
		\Else{}
			 {\bf return} $1+\sum_{c \in b.\ms{children}} \lit{num-desc}(c)$ \Comment{recurse at children}
		\EndIf \label{line:num-desc-end}
	}\EndPart
    }
    \algstore{pow3}
  \end{algorithmic}
\end{algorithm}


\subsection{Decided blocks and committed transactions}\label{ssec:termination}

A blockchain system $S$ must define when the block at an index is agreed upon.
To this end, it has to define a point in its execution where a prefix of the main branch can be ``reasonably'' considered
as persistent.\footnote{In theory, there cannot be consensus on a block at a partiular index~\cite{FLP85}, hence preventing persistence, 
however, applications have successfully used Ethereum to transfer 
digital assets based on parameter $m_{\ms{ethereum}}=11$~\cite{NG16}.
}
More precisely, there must exist a parameter $m$ provided by $S$
for an application to consider a block as \emph{decided} and its transactions as \emph{committed}. This parameter is typically $m_{\ms{bitcoin}} = 5$ in Bitcoin (Alg.~\ref{alg:nakamoto}, line~\ref{line:m-btc}) and $m_{\ms{ethereum}} = 11$ in Ethereum (Alg.~\ref{alg:ghost}, line~\ref{line:m-eth}).

\begin{definition}[Transaction commit]
Let $\ell_i = \tup{B_i, P_i}$ be the blockchain view at node $p_i$ in system $S$.
For a transaction $\ms{tx}$ to be \emph{locally committed} at $p_i$, the conjunction of the following properties must hold in $p_i$'s view $\ell_i$: 
\begin{enumerate}
\item Transaction $\ms{tx}$ has to be in a block $b_0 \in B_i$ of the main branch of system $S$. 
Formally, $\ms{tx}\in b_0 \wedge b_0 \in B'_i : c_i = \tup{B'_i, P'_i} = \lit{get-main-branch}()_i$. 
\item There should be a subsequence of $m$ blocks $b_1, ..., b_m$ appended after block $b$. Formally, $\exists b_1, ..., b_m \in B_i : \tup{b_1, b_0}, \tup{b_{2}, b_{1}}, ..., \tup{b_m, b_{m-1}} \in P_i$. (In short, we say that $b_0$ is \emph{decided}.)
\end{enumerate}
A transaction $\ms{tx}$ is \emph{committed} if there exists a node $p_i$ such that $\ms{tx}$ is \emph{locally committed}.
\end{definition}

Property (1) is needed because nodes eventually agree on the main branch that defines the 
current state of accounts in the system---blocks that are not part of the main branch are ignored.
Property (2) is necessary to guarantee that the blocks and transactions currently in the main branch
will persist and remain in the main branch.
Before these additional blocks are created, nodes may not have reached consensus regarding the 
unique blocks $b$ at index $j$ in the chain. This is illustrated by the fork of Figure~\ref{fig:chains}
where nodes consider, respectively, the pointer 
$\tup{b_1, g}$ and the pointer
$\tup{b_2, g}$ in their local blockchain view.
By waiting for $m$ blocks were $m$ is given by the blockchain system, the system guarantees 
with a reasonably high probability that nodes will agree on the same block $b$.
%


\begin{algorithm}[!ht]
  \caption{Checking transaction commit at node $p_i$}\label{alg:commit}
  \begin{algorithmic}[1]
    \algrestore{pow3}
    {\footnotesize

	\Part{$\act{is-committed}(\ms{tx})_i$}{ \label{line:pow-commit-starts} \Comment{check whether transaction is committed}
	        \State $\tup{B'_i, P'_i} \gets \act{get-main-branch}()$ \Comment{pick main branch with Alg.~\ref{alg:nakamoto} or~\ref{alg:ghost}}\label{line:get-main-branch}
	        \If{$\exists b_0 \in B'_i :  \ms{tx} \in b_0 \,\wedge\, \exists b_1, ..., b_m \in B_i :$ \Comment{\ms{tx} in main branch}
	        		\State \T$ \tup{b_{1}, b_0}, \tup{b_2,b_1} ..., \tup{b_{m},b_{m-1}} \in P_i$}   \Comment{enough blocks}
				\State {\bf return} {$\lit{true}$} 
			\Else{} {\bf return} {$\lit{false}$}
	        \EndIf 
	        \label{line:pow-commit-ends}
	}\EndPart 
	
    }
  \end{algorithmic}
\end{algorithm}

For example, consider a fictive blockchain system with $m_{\ms{fictive}}=2$ that selects the heaviest branch (Alg.~\ref{alg:ghost}, lines~\ref{line:eth-get-branch-start}--\ref{line:eth-get-branch-end}) as its main branch. If the blockchain state was the one depicted in Figure~\ref{fig:reorg}, then blocks $b_2$ and $b_5$ would be decided and all their transactions would be committed. This is because they are both part of the main branch and they are followed by at least 2 blocks, $b_8$ and $b_{13}$. (Note that we omit the genesis 
block as it is always considered decided but does not include any transaction.)

\section{The Balance Attack}\label{sec:pa}

In this section, we present the Balance attack, a novel form of attacks that affect proof-of-work blockchains, especially Ethereum.
Its novelty lies in identifying subgroups of miners of equivalent mining power and delaying messages between them rather than entering a race to mine blocks faster than others.

The balance attack demonstrates a fundamental limitation of main proof-of-work systems in that they are block oblivious.
%
\begin{definition}[Block Obliviousness]\label{def:obliviousness}
A blockchain system is \emph{block oblivious} if an attacker can:
\begin{enumerate}
\item make the recipient of a transaction $\ms{tx}$ observe that $\ms{tx}$ is committed and 
\item later remove the transaction $\ms{tx}$ from the main branch,
\end{enumerate}
with a 
probability $1-\varepsilon$, where $\varepsilon$ is a small constant.
\end{definition}

The balance attack is simple: after the attacker \vincent{disrupts communication} introduces a delay between correct 
subgroups of equivalent mining power, it simply issues transactions in one subgroup. The 
attacker then mines sufficiently many blocks in another subgroup to ensure with high 
probability that 
the subtree of another subgroup outweighs the transaction subgroup's. 
Even though the transactions are 
committed, the attacker can rewrite with high probability the blocks that contain these transactions by 
outweighing the subtree containing this transaction. 

Note that one could benefit from delaying messages only between the merchant and the rest of the 
network by applying the eclipse attack~\cite{HKZG15} to Ethereum.  
Eclipsing one node of Bitcoin appeared, however, sufficiently difficult: 
it requires to restart the node's protocol in order to control all the logical neighbors the node will eventually try to connect to. While a Bitcoin node typically connects to 8 logical 
neighbors, an Ethereum node typically connects to 25 nodes, making the problem even
harder. Another option would be to isolate a subgroup of smaller 
mining power than another subgroup, however, it would make the attack only possible if the 
recipients of the transactions are located in the subgroup of smaller mining power. Although possible
this would limit the generality of the attack, because the attacker would be constrained on  
the transactions it can override.

Note that the Balance Attack inherently violates the persistence of the main branch prefix
and is enough for the attacker to double spend.
The attacker has simply to identify the subgroup that contains merchants and create transactions to buy goods from these merchants. After that, it can issue the transactions to this subgroup while propagating its mined blocks to at least one of the other subgroups.
Once the merchant shipped goods, 
the attacker stops delaying messages.
Based on the high probability that the tree seen by the merchant is outweighed by another subtree, 
the attacker could reissue another transaction transferring the exact same coin again.
%
%
%

\begin{figure}[t]\
\begin{center}
\subfigure[Example of the selection of edges $E_0$ delayed between $k=2$ subgraphs of 25 units of mining power each\label{sfig:partition}]{\includegraphics[scale=0.41, clip=true, viewport=20 50 300 350]{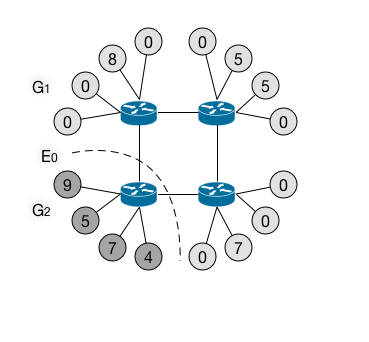}}\hspace{0.5em}
\subfigure[Example of the selection of edges $E_0$ delayed between $k=4$ subgraphs of 12 units of mining power each\label{sfig:partition2}]{\includegraphics[scale=0.41, clip=true, viewport=100 50 400 350]{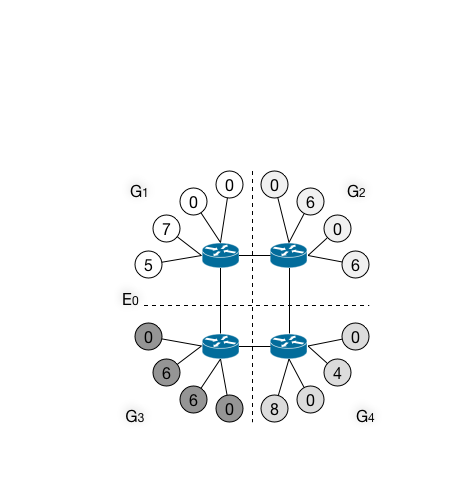}}
\caption{Two decompositions of communication graphs into subgraphs by an attacker where $E_0$ represents 
the cut of communication edges linking the subgraphs}
\end{center}
\end{figure}

\subsection{Executing a Balance Attack}
For the sake of simplicity, let us fix $k=2$ and postpone the general analysis for any $k\geq 2$ to Appendix~\ref{app:proof}. 
We consider subgraphs $G_1 = \tup{V_1, E_1}$ and $G_2 = \tup{V_2, E_2}$ of the communication graph 
$G = \tup{V,E}$ so that each subgraph has half of the mining power of the system as depicted in Figure~\ref{sfig:partition} whereas Figure~\ref{sfig:partition} illustrates the variant when $k=4$.
Let $E_0 = E \setminus (E_1 \cup E_2)$ be the set of edges that connects nodes of $V_1$ to nodes of 
$V_2$ in the original graph $G$. 
Let $\tau$ be the communication delay introduced by the attacker on the edges of $E_0$.

As indicated in Algorithm~\ref{alg:paa}, the attacker can introduce a sufficiently long delay $\tau$ during which the miners of $G_1$ mine in isolation of the miners of $G_2$ (line~\ref{line:tau}). As a consequence, different transactions get committed
in different series of blocks on the two blockchains locally viewed by the subgraphs $G_1$ and $G_2$. 
Let $b_2$ be a block present only in the blockchain viewed by $G_2$ but absent 
from the blockchain viewed by $G_1$.
In the meantime, the attacker issues transactions
spending coins $C$ in $G_1$ (line~\ref{line:spend}) and mines a blockchain starting from the block $b_2$ (line~\ref{line:mineb2}).
Before the delay expires the attacker sends his blockchain to $G_2$.
After the delay expires, 
the two local views of the blockchain are exchanged. 
Once the heaviest branch that the attacker contributed to is adopted, the attacker
can simply reuse the coins $C$ in new transactions (line~\ref{line:doublespend}).

\begin{algorithm}[t]
  \caption{The Balance Attack initiated by attacker $p_i$}\label{alg:paa}
  \begin{algorithmic}[1]
    {\footnotesize

	\Part{State}{ \label{line:bitcoin-start}
	  \State $G = \tup{V,E}$, the  communication graph
	  \State $\lit{pow}$, a mapping from of a node in $V$ to its mining power in ${\mathbb R}$
	  \State $\ell_i = \tup{B_i, P_i}$, the local blockchain at node $p_i$ is a directed acyclic 
	  \State \T\T graph of blocks $B_i$ and pointers $P_i$
	  \State $\rho \in (0;1)$, the portion of the mining power of the system owned by 
	  \State \T\T the attacker $p_i$, typically $\rho < 0.5$
	  \State $d$, the difficulty of the crypto-puzzle currently used by the system
	}\EndPart
	 
	\Statex
	
	\Part{$\act{balance-attack}(\tup{V,E})_i$}{ \label{line:attack-starts} \Comment{starts the attack}
	 \State Select $k\geq 2$ subgraphs $G_1 = \tup{V_1,E_1}, ..., G_k= \tup{V_k,E_k}$: 
	 \State \T\T$\sum_{\forall v \in V_1} \lit{pow}(v) \approx ... \approx \sum_{\forall v' \in V_k} \lit{pow}(v')$
	 \State Let $E_0 = E \setminus \cup_{\forall 0<i\leq k} E_i$ \Comment{attack communication channels}
	 \State Stop communications on $E_0$ during
	 $\tau \geq \frac{(1-\rho)6d\log(\frac{4}{\varepsilon})}{4\rho^2t}$ seconds\label{line:tau}
	 \State Issue transaction $\ms{tx}$ crediting a merchant in graph $G_i$ with coins $C$\label{line:spend}
	 \State Let $b_2$ be a block appearing in $G_j$ but not in $G_i$ \label{line:b2}
	 \State Start mining on $\ell_i$ immediately after $b_2$\label{line:mineb2}
	 \Comment{contributed to correct chain}
	 \State Send blockchain view $\ell_i$ 
	 to some subgraph $G_j$ where $j\neq i$	
	 \State When $\tau$ seconds have elapsed, stop delaying communications on $E_0$\label{line:stop}
	 \State Issue transaction $\ms{tx}'$ that double spends coins $C$\label{line:doublespend}
	}\EndPart
	
    }
  \end{algorithmic}
\end{algorithm}

\subsection{The knowledge about the network}

As indicated by the state of Algorithm~\ref{alg:paa}, an attacker has to be knowledgeable about the current 
settings of the blockchain system to execute a Balance attack. In fact, the attacker must have information regarding the 
logical or physical communication graph, the mining power of the miners or pools of miners and the current 
difficulty of the crypto-puzzle. In particular, this information is needed to delay messages for long enough 
between 
selected communication subgraphs. As we will show in the next section, this delay can be overestimated, 
however, underestimating it may make the attack impossible.

The information regarding the mining power and the difficulty of nodes is generally public information and can 
often be retrieved online. 
In particular, we got provided an access to the R3 Ethereum network statistics \url{http://r3n1-utils.gcl.r3cev.com/} that included block 
propagation delays, the list of connected nodes and their individual mining power, the version of their Ethereum
client as well as the current difficulty of the crypto-puzzle. 
The same statistical information regarding the Ethereum public chain is publicly available online at  \url{https://ethstats.net/}. 

It is more difficult, however, to gather information regarding the communication network. 
Note that some tools exist to retrieve information regarding the communication topology of blockchain systems. 
The interesting aspects of the Balance attack is that it can apply to the logical overlay used by the peer-to-peer network of the blockchain system or to the physical overlay. While there exist tools to retrieve the logical overlay 
topology, like  AddressProbe~\cite{MLP+15} to find some information regarding the peer-to-peer overlay of Bitcoin, it can be easier for an attacker of Ethereum to run a DNS poisoinning or a denial-of-service attack rather than a BGP hijacking~\cite{LS14,AZV16} that requires access to autonomous systems.


\section{Vulnerability of the {\sc Ghost} Protocol}\label{sec:proof}

In this Section, we show that the Balance attack makes a blockchain system based on {\sc Ghost} (depicted in Alg.~\ref{alg:ghost}) block oblivious. 
A malicious user can issue a Balance attack with less than half of the mining power by delaying the network. 
Let us first summarize previous and new notations in Table~\ref{table:analysis}.

\begin{table}[ht!]
\begin{tabular}{r|l}
t & total mining power of the system (in million hashes per second, MH/s)\\
d & difficulty of the crypto-puzzle (in million hashes, MH)\\
$\rho$ & fraction of the mining power owned by the malicious miner (in percent, \%)\\
k & the number of communication subgraphs \\
$\tau$ & time during which communication between subgraphs is disrupted (in seconds, s) \\
$\mu_c$ & mean of the number of blocks mined by each communication subgraph during $\tau$ \\
$\mu_m$ & mean of the number of blocks mined by the attacker during time $\tau$ \\
$\Delta$ & the maximum difference of mined blocks for the two subgraphs
\end{tabular}
\caption{Notations of the analysis\label{table:analysis}}
\end{table}

For the sake of simplicity in the proof, we assume that $k=2$ and  
$\sum_{\forall v \in V_1} \lit{pow}(v) =  \sum_{\forall v' \in V_2} \lit{pow}(v')$ so that 
the communication is delayed between only two communication subgraphs of equal mining power.
We defer the proof for the general case where $k\geq 2$ to Appendix~\ref{app:proof}.

\begin{figure*}[t]
\begin{center}
\includegraphics[scale=0.65]{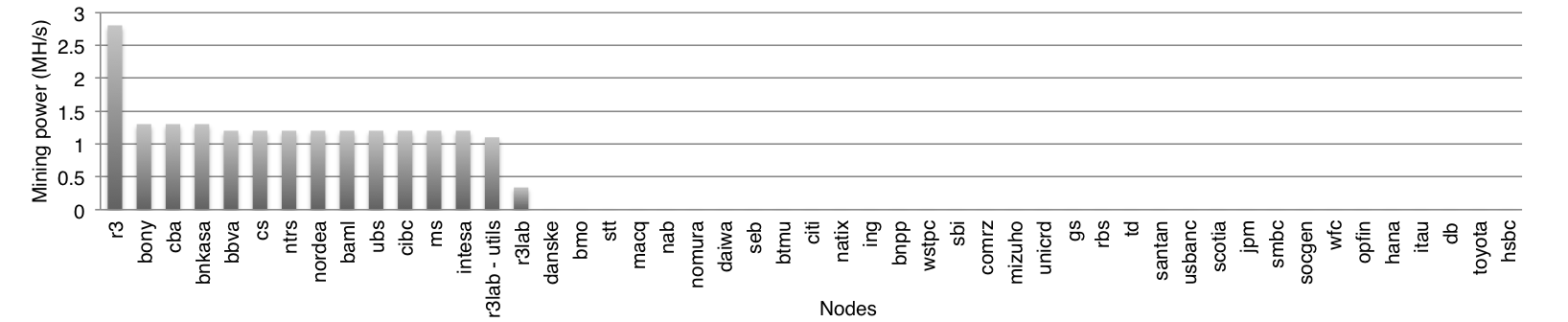}
\caption{The mining power of the R3 Ethereum network as reported by \texttt{eth-netstats} as of June 2016\label{fig:mining-power}}
\end{center}
\end{figure*}

As there is no better strategy for solving the crypto-puzzles than random trials, we consider that 
subgraphs $G_1$ and $G_2$ mine blocks during delay $\tau$.
During that time, each of $G_1$ and $G_2$ performs a series of $n = \frac{1-\rho}{k}t\tau$ independent and identically distributed Bernoulli trials that returns one in case of success with probability 
$p = \frac{1}{d}$ and 0 otherwise. 
Let the sum of these outcomes for subgraphs $G_1$ and $G_2$ be the random variables $X_1$  and $X_2$, respectively,
each with a binomial distribution and mean: 
\begin{equation}
\mu_c = np = \frac{(1-\rho)t\tau}{2d}. \label{eq:mean}
\end{equation}

Similarly, the mean of the number of blocks mined by the malicious miner during time $\tau$ is
\begin{equation}
\mu_m = \frac{\rho t\tau}{d}. \notag 
\end{equation}


From line~\ref{line:tau} of Alg.~\ref{alg:paa}, we know that $\tau \geq 
\frac{(1-\rho)6d\log(\frac{4}{\varepsilon})}{4\rho^2t}$
which leads to:
\begin{eqnarray}
\tau &\geq& \frac{(1-\rho)6d\log(\frac{4}{\varepsilon})}{4\rho^2t}, \notag \\
\frac{(1-\rho)t\tau}{2d} &\geq& \frac{3(1-\rho)^2\log(\frac{4}{\varepsilon})}{4\rho^2}. \notag
\end{eqnarray}
By Eq.~\ref{eq:mean} we have:
\begin{eqnarray}
\mu_c &\geq& \frac{3(1-\rho)^2\log(\frac{4}{\varepsilon})}{4\rho^2}, \notag \\
\frac{4\rho^2 \mu_c}{3(1-\rho)^2} &\geq& \log\left(\frac{4}{\varepsilon}\right), \notag \\
-\frac{4\rho^2 \mu_c}{3(1-\rho)^2} &\leq& \log\left(\frac{\varepsilon}{4}\right), \notag \\
e^{-\frac{4\rho^2 \mu_c}{3(1-\rho)^2}} &\leq& \frac{\varepsilon}{4}, \notag \\
1-4e^{-\frac{4\rho^2 \mu_c}{3(1-\rho)^2}} &\geq& 1-\varepsilon.   \label{eq:epsilon}
\end{eqnarray}

The attack relies on minimizing the difference in mined blocks between any pair of communication subgraphs. 

\begin{lemma}\label{lem:delta}
After the communication is re-enabled,
the expectation of the number of blocks mined by the attacker is greater than the difference $\Delta = |X_1 - X_2|$ of the number of blocks mined on the two subgraphs $G_1$ 
and $G_2$, with probability $1-\varepsilon$. 
\end{lemma}
\begin{proof}
The communication is re-enabled at line~\ref{line:stop} of Alg.~\ref{alg:paa}.
At this point in the execution, the probability that the numbers of blocks mined by each subgraph are within a $\pm \delta$ factor from their mean, is bound for $0 < \delta < 1$ and 
$i\in \{1,2\}$
by Chernoff bounds~\cite{MR95}: 
\begin{equation}
\left\{\begin{array}{ll}
\Pr\left[X_i \geq (1+\delta)\mu_c\right] &\leq e^{-\frac{\delta^2}{3}\mu_c},\notag
\\
\Pr\left[X_i \leq (1-\delta)\mu_c\right]& \leq e^{-\frac{\delta^2}{2}\mu_c}. 
\end{array}\right.
\end{equation}
Thus, we have 
\begin{eqnarray}
\Pr[|X_i -\mu_c| < \delta \mu_c] &>& 1-2e^{-\frac{\delta^2}{3}\mu_c}. \notag
\end{eqnarray}

Observe that the probability that these two random variables 
are both within a $\pm \delta\mu_c$ is lower than the probability that their 
difference $\Delta$ is upper-bounded by $2\delta \mu_c$:
\begin{equation}
\left(\Pr\left[|X_i - \mu_c| < \delta \mu_c\right]\right)^2 \leq \Pr[\Delta < 2\delta \mu_c].\notag
\end{equation}
Thus, we obtain:
\begin{equation}
\Pr[\Delta < 2\delta \mu_c] > \left(1-2e^{-\frac{\delta^2}{3}\mu_c}\right)^2.\label{eq:power-2}
\end{equation}

As $\mu_c \geq 0$, we have that:
\begin{eqnarray}
{-\frac{\delta^2}{3}\mu_c} &\leq& 0, \notag \\
-e^{\frac{-\delta^2}{3}\mu_c} &\geq& -1, \notag
\end{eqnarray}
and we can apply the Bernoulli inequality to Eq.~\ref{eq:power-2}:
\begin{eqnarray}
\Pr[\Delta < 2\delta \mu_c] &>& 1-4e^{-\frac{\delta^2}{3}\mu_c}. \notag 
\end{eqnarray}

If we replace $\delta$ by $\frac{2\rho}{1-\rho}$, we obtain:
\begin{eqnarray}
\Pr[\Delta < \mu_m] &>& 1-4e^{-\frac{4\rho^2 \mu_c}{3(1-\rho)^2}}. \notag 
\end{eqnarray}
By Eq.~\ref{eq:epsilon}, we can see that the expectation of the number of blocks mined by the attacker is strictly 
greater than $\Delta$ with high probability:
\begin{eqnarray}
\Pr[\Delta < \mu_m] &>& 1-\varepsilon.\notag
\end{eqnarray}
\end{proof}


\begin{theorem}\label{thm:block-obliviousness}
A blockchain system that selects a main branch based on the {\sc Ghost} protocol (Alg.~\ref{alg:ghost}) is block oblivious.
\end{theorem}
\begin{proof}
By lines~\ref{line:num-desc-start}--\ref{line:num-desc-end} of Alg.~\ref{alg:ghost}, we know that {\sc Ghost} 
counts every mined blocks to compute the weight of a subtree, and to select one blockchain view and discard 
the other. 

Since the expected number of blocks mined by the attacker is greater than the difference $\Delta$ with  
probability $1-\varepsilon$, we know that when the timer expires at line~\ref{line:stop} of Alg.~\ref{alg:paa}, the 
attacker will make the system discard the blockchain view of either $G_1$ or $G_2$ by sending its blockchain 
view to the other subgraph, hence making the blockchain system block oblivious.
\end{proof}


%
%

%
%

\section{Analysis of the R3 Testbed}\label{sec:example}

The statistics of the R3 testbed were gathered through the \texttt{eth-netstat} applications at the end of June 2016.\footnote{\url{http://r3n1-utils.gcl.r3cev.com/}.}
R3 is a consortium of more than 50 banks that has tested blockchain systems and in particular Ethereum in a consortium private chain context over 2016.\footnote{\url{http://www.coindesk.com/r3-ethereum-report-banks}}
As depicted in Figure~\ref{fig:mining-power}, the network consisted at that time of $|V| = 50$ nodes among which only 15 were mining. The mining power of the system was about $20\,$MH/s and the most powerful miner mines at 2.4\,MH/s or 12\% of the total mining power while the difficulty of the crypto-puzzle was observed close to 30\,MH.  

\begin{table}[ht!]
\begin{center}
\begin{tabular}{r|l}
Number of nodes & 50 \\
Number of miners & 15 \\
Total mining power (MH/s) & 20 \\
Mining power of the most powerful miner (MH/s) & 2.4 \\
Difficulty (MH) & 30 
\end{tabular}
\caption{The R3 settings used in the analysis\label{table:r3}}
\end{center}
\end{table}

%

We assume that the attacker selects communication edges $E_0$ between two subgraphs $G_1$ and $G_2$ so that their mining power is $8.8\,$MH/s each.
The probability $p$ of solving the crypto-puzzle per hash tested is $\frac{1}{30\times10^{6}}$ so that the mean 
is $\mu_c = np = 95.3$ if we wait for 19 minutes and 40 seconds. 
The attacker creates, in expectation,
a block every 
$\frac{30}{2.4}=12.5$ seconds or $\lfloor\frac{1180}{12.5}\rfloor = 94$ blocks during the 19 minutes and 
40 seconds.
Hence let us select $\delta$ such that the attacker has a chance to mine more than $2\delta\mu_c$ blocks during that period, i.e., $94 = 2\delta\mu_c+1$ implying that $\delta = 0.1343$.
The probability of the difference exceeding $94$ is upper bounded by $4e^{-\frac{\delta^2}{3}\mu_c}$ 
leading to a bound of 49.86\%.

To conclude a single miner of the R3 testbed needs to delay edges of $E_0$ during less than 20 minutes to execute a Balance attack and discards blocks that were previously decided (even if $m=18$ was chosen) with a chance of success greater than $\frac{1}{2}$.

\subsection{Tradeoff between communication delays and mining power}
Malicious nodes may have an incentive to form a coalition in order to exploit the Balance attack to double spend. In this case, it is easier for the malicious nodes to control a larger portion of the mining 
power of the system, hence such a coalition would not need to delay messages as long as in our example of Section~\ref{sec:example}. For simplicity, we again consider the case where the attacker delay messages between $k=2$ communication subgraphs.

To illustrate the tradeoff between communication delay and the portion of the mining power controlled by the attacker, we consider the R3 testbed with a 30\,MH total difficulty, a 20\,MH/s total mining power and plot the 
probability as the communication delay increases for different portions of the mining power controlled by the adversary.
Figure~\ref{fig:proba-vs-power} depicts this result.  As expected, the probability increases exponentially fast as the delay increases, and the higher the portion of the mining power is controlled by the adversary the faster the probability increases. In particular, in order to issue 
a balance attack with 90\% probability, 51 minutes are needed for an adversary controlling 12\% of the total mining power whereas only 11 minutes are sufficient for an adversary who controls 20\% of the mining power. 
\vincent{Look at the time needed for 60\% or 40\% of the mining power.}

\begin{figure}[t]
\begin{center}
\includegraphics[scale=0.55]{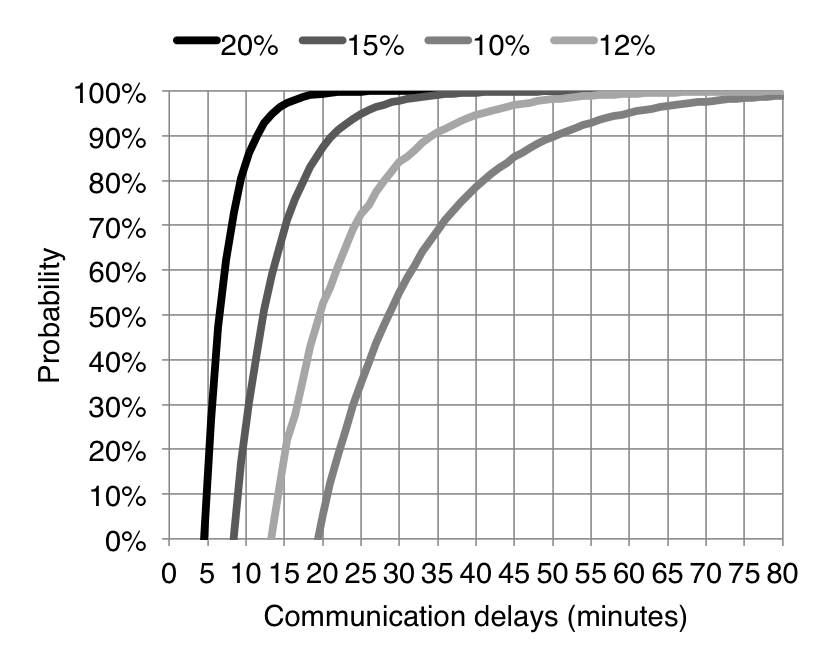}
\caption{Probability of Balance attack in the R3 testbed as the communication delay increases for different portions of the mining power controlled by the attacker\label{fig:proba-vs-power}}
\end{center}
\end{figure}


%

\subsection{Tradeoff between communication delays and difficulties}

Another interesting aspect of proof-of-work blockchain is the difficulty parameter $d$. As already mentioned, this parameter impacts the 
expected time it takes for a miner to succeed in solving the crypto-puzzle. When setting up a private chain, one has to choose a difficulty to make sure the miners would mine at a desirable pace. A too high difficulty reduces the throughput of the system without requiring leader election~\cite{EGSvR16} or consensus sharding~\cite{LNZ16}.
A too low difficulty increases
the probability for two correct miners to solve the crypto-puzzle before one can propagate the block to the other, a problem of Bitcoin that motivated the {\sc Ghost} protocol~\cite{SZ15}. 

\begin{figure}[ht!]
\begin{center}
\includegraphics[scale=0.55]{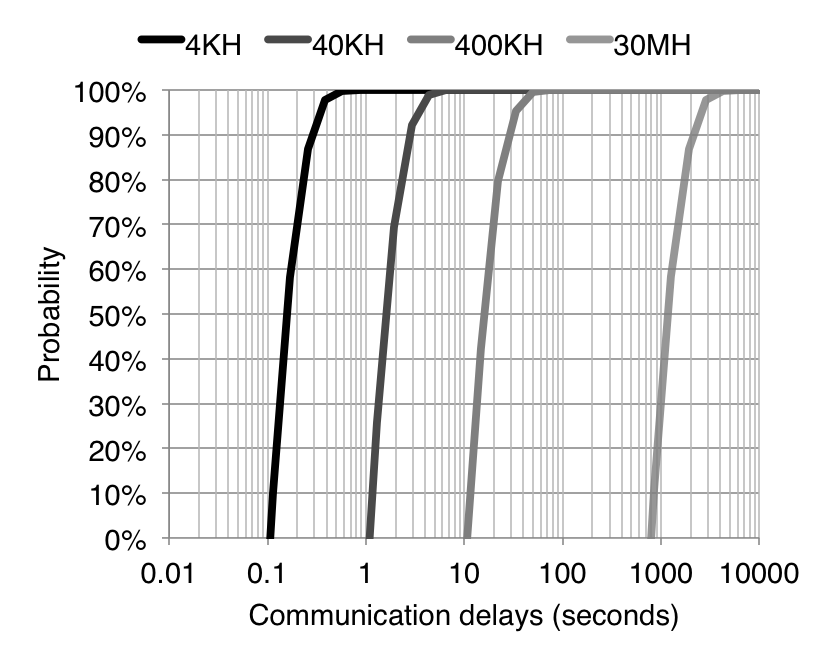}
\caption{Probability of Balance attack in the R3 testbed (with 20\,MH/s of total mining power, 12\% of the mining power at the attacker) for different difficulties as the communication delay increases\label{fig:proba-vs-difficulty}}
\end{center}
\end{figure}

Figure~\ref{fig:proba-vs-difficulty} depicts the probability of the Blockchain anomaly when the communication delay increases for different difficulties without considering the time for a block to be decided. 
Again, we consider the R3 Ethereum network with a total mining power of 30\,MH/s and an attacker owning 
$\rho = 12\%$ of this mining power and delaying communications between $k=2$ subgraphs of half of the remaining mining power ($\frac{1-\rho}{2} = 44\%$) each.
The curve labelled 4\,KH indicates a difficulty of 4000 hashes, which is also the difficulty chosen by default by Ethereum when setting up a new private chain system. This difficulty is dynamically adjusted by Ethereum at runtime to keep the mining block duration constant in expectation, however, this adaptation is dependent on the visible mining power of the system. The curve labelled 30\,MH indicates the probability for the difficulty observed in the R3 Ethereum network. We can clearly see that the difficulty impacts the probability of the Balance attack. This can be explained by the fact that the deviation of the random variables $X_1, ..., X_k$ from their mean $\mu_c$ is bounded for sufficiently large number of mined blocks.



\section{Experimenting the Balance Attack on an Ethereum Private Chain}\label{sec:expe}

In this section, we experimentally produce the attack on an Ethereum private chain involving up to 18 physical 
distributed machines. 
To this end, we configure a realistic network with 15 machines dedicated to mining as in the R3 Ethereum 
network
we described in Section~\ref{sec:example} and 3 dedicated network switches.

All experiments were run on 18 physical machines of the Emulab environment where a network topology was configured using ns/2 as depicted in Figure~\ref{fig:emulab1}.  The topology consists of three local area networks configured through a ns/2 configuration file with 20\,ms latency and 100\,Mbps bandwidth. 
All miners run the \texttt{geth} Ethereum client v.1.3.6 and the initial difficulty of the crypto-puzzle is set 
to 40\,KH. The communication graph comprises the subgraph $G_1$ of 8 miners that includes the attacker and 7 correct miners and a subgraph $G_2$ of 7 correct miners. 


\begin{figure}
\begin{center}
\includegraphics[scale=0.4, clip=true, viewport=200 20 900 300]{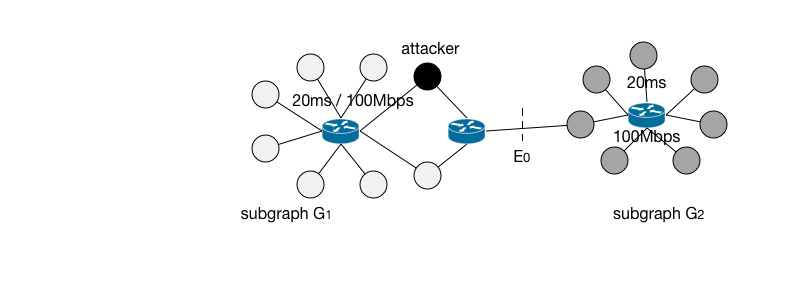}
\caption{The topology of our experiment involving 15 miners with subgraph $G_1$ including the attacker depicted in black and subgraph $G_2$ depicted in grey\label{fig:emulab1}}
\end{center}
\end{figure}

\subsection{Favoring one blockchain view over another}\label{ssec:selection}

We run our first experiment during 2 minutes.
We delayed the link $E_0$ during 60 seconds so that both subgraphs mine in isolation from each other during that time and end up with distinct blockchain views. After the delay we take a snapshot, at time $t_1$, 
of the blocks mined by each subgraphs and the two subgraphs start exchanging information normally leading to a consensus regarding the current state of the blockchain.  
At the end of the experiment, after 2 minutes we take another snapshot 
$t_2$ of the blocks mined by each subgraph. 

\begin{table}
\begin{center}
\includegraphics[scale=0.55,clip=true, viewport=0 50 400 750]{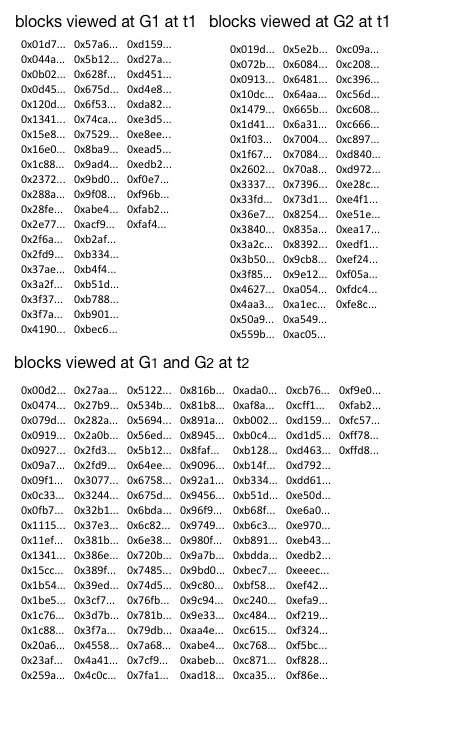}
\caption{Hash prefixes of the blocks of the main branch (excluding uncles) selected by the subgraphs $G_1$ and $G_2$; the blocks of $G_2$ at time $t_1$ do no longer appear at time $t_2$\label{table:block-hashes}}
\end{center}
\end{table}

Table~\ref{table:block-hashes} lists the prefix of the hashes of blocks (excluding uncles) of the blockchain views of $G_1$ and $G_2$ at times $t_1$, while the two subgraphs did not exchange their view, and at time $t_2$, after the subgraphs exchanged their blocks. Note that we did not represent the uncle blocks to focus on the main branches.
We observe that the blockchain view of the subgraph $G_1$ was adopted as the valid chain while the 
other blockchain view of the subgraph $G_2$ was not. 
For example, one of the listed blocks viewed by $G_1$ at time $t_1$ whose hash starts with \texttt{0xfab2} appears as well in the final blockchain at time $t_2$.
More generally, we can see that only blocks of the blockchain view of $G_1$ at time $t_1$ were finally adopted at time $t_2$ as part of the blocks of the global view of the blockchain. 
All the blocks of $G_2$ at time $t_1$ were discarded from the blockchain by time $t_2$.

\subsection{Blocks mined by an attacker and two subgraphs}

We now report the total number of blocks mined, especially focusing on the creation of uncle blocks.
More precisely, 
we compare the number 
of blocks mined by the attacker against the difference of the number of blocks $\Delta$ mined by each subgraph.
We know from the analysis that it is sufficient for the attacker to mine at least $\Delta + 1$ blocks in order to be 
able to discard one of the $k$ blockchain views, allowing for double spending. 
The experiment is similar to the previous experiment in that we also used Emulab with the same ns/2 topology, however, we did not introduce delays and averaged results over 10 runs of 4 minutes each.

Figure~\ref{fig:deltas} depicts the minimum, maximum and average blocks obtained over the 10 runs. The vertical bars indicate minimum and maximum. First, we can observe that the average difference $\Delta$ is usually close to its minimum value observed during the 10 runs. This is due to having a similar total number of blocks mined by each subgraph in most cases with few rare cases where the difference is larger.
As we can see, the total number of blocks (including uncles) mined during the experiment by the attacker is way larger than the difference in blocks $\Delta$ mined by the two subgraphs. This explains the success of the Balance attack
as was observed in Section~\ref{ssec:selection}.






\begin{figure}
\begin{center}
\includegraphics[width=0.4\textwidth]{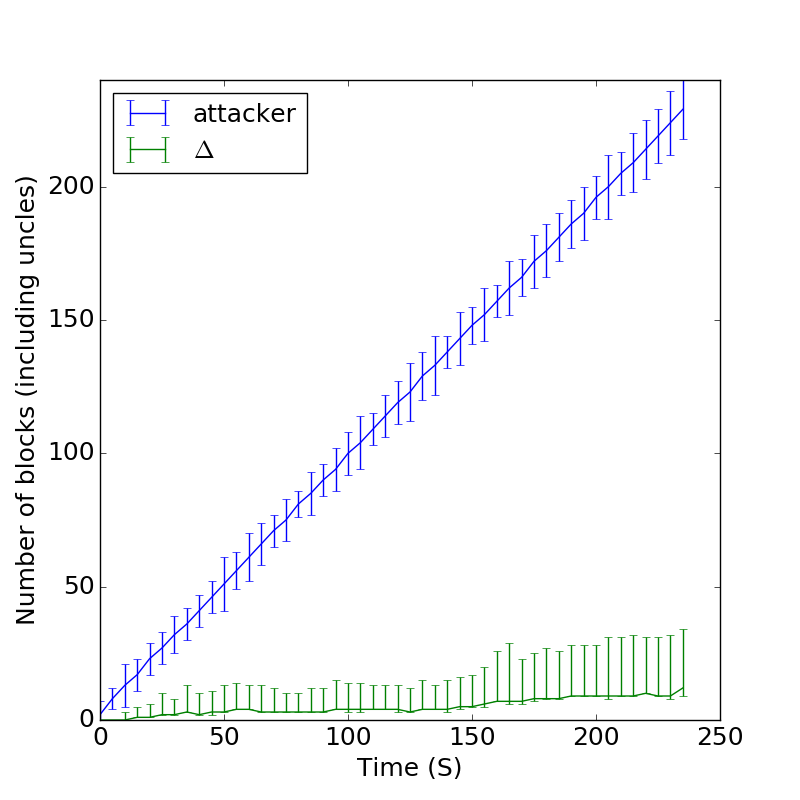}
\caption{The total number of blocks (including uncles) mined by the attacker and the difference $\Delta$ in the total number of blocks (including uncles) mined by the two subgraphs $G_1$ and $G_2$ (error bars indicate minimum and maximum observed over 10 runs)\label{fig:deltas}}
\end{center}
\end{figure}

\begin{figure}
\begin{center}
\includegraphics[width=0.4\textwidth]{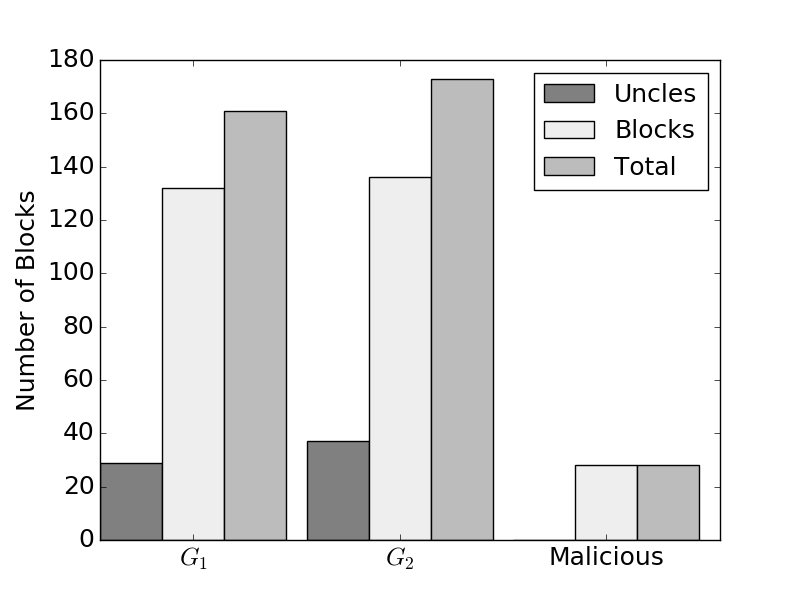}
\caption{The total number of blocks (including uncles and non-uncles) 
mined by the attacker and the two subgraphs $G_1$ and $G_2$\label{fig:bars}}
\end{center}
\end{figure}

\begin{figure}
\begin{center}
\includegraphics[width=0.4\textwidth]{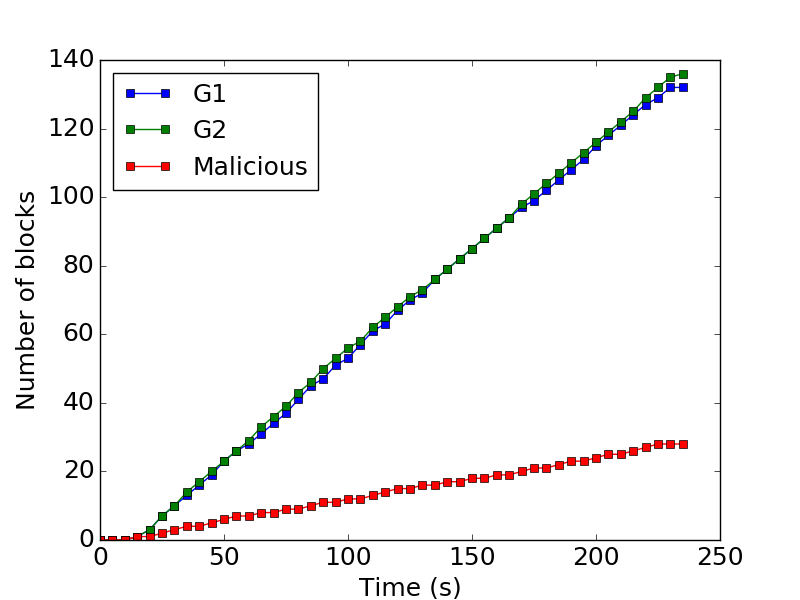}
	\caption{The depth of the blockchains mined by the attacker and two subgraphs $G_1$ and $G_2$\label{fig:sevensevenonlyblocks}}
\end{center}
\end{figure}

%
%

\subsection{The role of uncle blocks in Ethereum}

In the previous experiment, we focused on the total number of blocks without differentiating the blocks that 
are adopted in the main branch and the uncle blocks that are only part of the local blockchain views. The {\sc Ghost} protocol accounts for these uncle blocks to decide the current state of the blockchain as we explained previously in Section~\ref{sec:prel}.

Figure~\ref{fig:bars} indicates the number of uncle blocks in comparison to the blocks accepted on the current state of the blockchain for subgraphs $G_1$ and $G_2$, and the attacker (referred to as `Malicious').
As expected, we can observe that the malicious node does not produce any uncle block because he mines the block in isolation of the rest of the network, successfully appending each mined block consecutively  
to the latest block of its current blockchain view. We note several uncle blocks in the subgraphs, as correct miners may mine blocks concurrently at the same indices.

Figure~\ref{fig:sevensevenonlyblocks} depicts the creation of the number of mined blocks (excluding uncle blocks) over time for subgraphs $G_1$ and $G_2$, and the attacker (referred to as `Malicious'). As we can see the difference between the number of blocks mined on the subgraphs is significantly smaller than the number 
of blocks mined by the attacker. This explains why the Balance attack was observed in this setting.

\subsection{Relating connectivity to known blocks}
Figure~\ref{fig:seveneighttwin} illustrates the execution of two subgraphs resolving connectivity issues and adopting a chain. This experiment outlines one of the fundamental aspects of the balance attack, in which the chosen subgraph resolves the network delay and attempts reconnection with another subgraph. At this point, the subgraphs will initiate the consensus protocol and select the branch to adopt as the main branch. The experiment was set up with two subgraphs $G_1$ and $G_2$ where $\left | V_1\right | = \left | V_2 \right | = 7$. The attacker selects a subgraph and delays messages between this subgraph and another, enforcing an isolated mining environment. Once the delay is set, the attacker joins one of the subgraphs and begins to mine onto the current chain. The attacker then delays the messages until there is a sufficient amount of blocks mined onto the isolated blockchain for it 
to be adopted as the correct chain by the other subgraph. In this experiment, at $t=60\,s$, the delay between subgraphs is resolved, and the subgraphs maintain a connection. Upon reconnection, the subgraphs invoke the consensus protocol to select and adopt the correct chain. In this case, using the {\sc Ghost} protocol, the heaviest chain is selected for both subgraphs, meaning the chain mined by $G_1$ is chosen, to which 
the attacker contributed.

This result reveals that the adoption of a chosen blockchain is plausible, given that the attacker is able to sufficiently delay messages between subgraphs. 

\begin{figure}
\begin{center}
\includegraphics[width=0.4\textwidth]{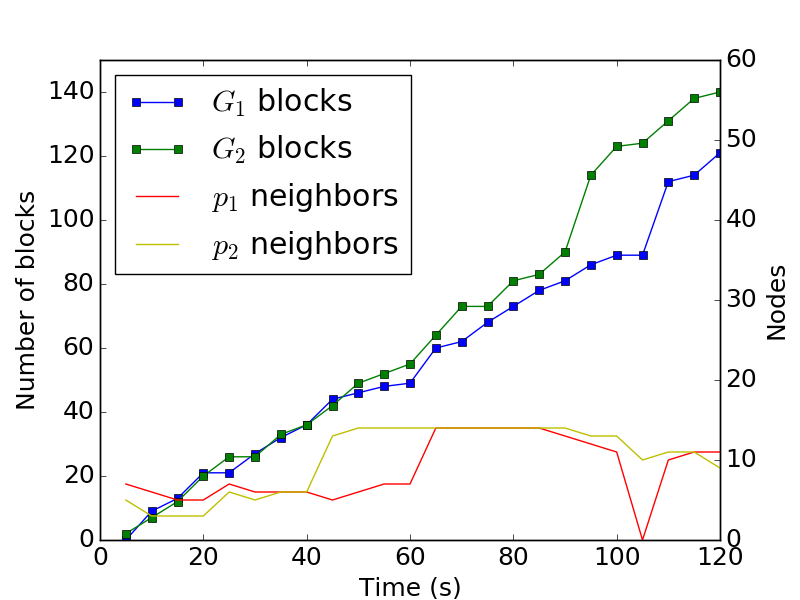}
	\caption{Execution where the attacker in $G_1$ delays the communication between $G_1$ and $G_2$ for one minute\label{fig:seveneighttwin}}
\end{center}
\end{figure}



\section{Application to Nakamoto's Consensus Protocol}\label{sec:disc}

Although the {\sc Ghost} protocol was shown to be less vulnerable to message delays than Nakamoto's consensus protocol~\cite{SZ15}, it is interesting to note that our analysis can be extended to Nakamoto's consensus protocol.
%
As described in Section~\ref{sec:prel}, Nakamoto's consensus protocol differs from the {\sc Ghost} protocol in the way a node chooses the current state of the blockchain in case of forks.
We present below an intuition of the proof that the Bitcoin is vulnerable to the Balance attack.

As noted previously~\cite{SZ15,PSS16} Nakamoto's protocol suffers from delays for the following reasons.
Let us assume that $p$ pools of miners mine $p$ concurrent blocks at the same time in Bitcoin. Assuming that 
all these mining pools share the same view of the blockchain with only a genesis block when this occurs, then all 
mining pools update their 
local blockchain view with a new block at index 1 that differs from one pool to another.
Finally, all miners exchange their blockchain view and selects one of these views as the current state of the blockchain.
Consider that the process repeats for the $f^{th}$ time, where miners mine concurrently, now at index $f$, and exchange the blockchain view. 
At each iteration of the process, the chain depth increases by 1 while the number of blocks mined effectively by the correct miners is $p$ leading to a final blockchain of depth $f$ while the number of mined blocks is $fp$.
Intuitively, this means that for an external attacker to be able to make the system discard a particular chain, the attacker simply needs to have slightly more than a $p^{th}$ 
of the total mining power or the other miners to have a chance to mine a longer chain than the other miners.

\begin{theorem}\label{thm:block-obliviousness}
A blockchain system that selects a main branch based on Nakamoto's protocol (Alg.~\ref{alg:nakamoto}) is block oblivious.
\end{theorem}

To show that the same result holds for Bitcoin, we need to slightly change the Balance attack.
While in Ethereum it was sufficient for the attacker to mine on any branch of the blockchain view of 
$G_j$ after 
the block $b_2$ (Alg.~\ref{alg:paa}, line~\ref{line:b2}), in Bitcoin the attacker has to mine 
at the top of the blockchain view of $G_j$. By doing so, the attacker increases the length of the 
Nakamoto's main branch in graph $G_j$. 
Considering that each correct miner mines at the top of the longest branch of their 
subgroup with the same probability $\pi$, 
the mean of the number of blocks added to the main chain will become 
$\mu^{bitcoin}_c = \frac{(1-\rho)t\tau}{2d\pi}$.
We can then define two binomial random variables $X'_i$ and $X'_j$ for the 
expected number of blocks in the main branch of $G_i$ and $G_j$, respectively, and apply the 
reasoning of the proof of Lemma~\ref{lem:delta}. 

\section{Related work}\label{sec:rw}

Traditional attacks against Bitcoin consist of waiting for some external action, like 
shipping goods, in response to a transaction before discarding the transaction from the main branch. 
As the transaction is revoked, the issuer of the transaction can reuse the coins of the 
transaction in another transaction. As the side effects of the external action cannot 
be revoked, the second transaction appears as a ``double spending''.

Perhaps the most basic form of such an attack assumes that an application takes 
an external action as soon as a transaction is included in a block~\cite{Fin11,KAC12,BDEWW13}.
The first attack of this kind is called Finney's attack and consists of solo-mining a block with a 
transaction that sends coins to itself without broadcasting it before issuing a transaction that double-spends the same coin to a merchant. When the goods are delivered in exchange of the coins, the attacker broadcasts its block to override the payment of the merchant. 
The vector76 attack~\cite{Vec11} consists of an attacker solo-mining after block $b_0$ a new block 
$b_1$ containing a transaction to a merchant to purchase goods. 
Once another block $b_1'$ is mined after $b_0$, the attacker quickly sends $b_1$ to the merchant for
an external action to be taken. If $b_1'$ is accepted by the system, the attacker can issue another transaction with the coins spent in the discarded block $b_1$.

The attacks become harder if the external action is taken after the transaction is committed
by the blockchain. Rosenfeld's attack~\cite{Ros12} consists of issuing a transaction to a merchant. 
The attacker then starts solo-mining a longer branch while waiting for $m$ blocks to be appended 
so that the merchant takes an external action in response to the commit.  
The attack success probability depends on the number $m$ of blocks the merchant waits before taking an external action and the attacker mining power. However, when the attacker has more mining power
than the rest of the system, the attack, also called \emph{majority hashrate attack} or \emph{51-percent attack}, is guaranteed successful, regardless of the value $m$. 
To make the attack successful when the attacker owns only a quarter of the mining power, the attacker can incentivize other miners to form a coalition~\cite{EG14} until the coalition owns more than half of the total mining power.

Without a quarter of the mining power, discarding a committed transaction in Bitcoin 
requires additional power, like the control over the network.
It is well known that delaying network messages can impact Bitcoin~\cite{DW13,PSS16,SZ15,GKK16,NKMS16}.
Decker and Wattenhoffer already observed that Bitcoin suffered 
from block propagation delays~\cite{DW13}. 
Godel et al.~\cite{GKK16} analyzed the effect of propagation delays on Bitcoin using a Markov 
process. 
Garay et al.~\cite{GKL15} investigated Bitcoin in the synchronous communication setting, however, 
this setting is often considered too restrictive~\cite{Cac01}.
Pass et al. extended the analysis for when the bound on message delivery is unknown
and showed in their model that the difficulty of Bitcoin's crypto-difficulty has to be adapted depending on the bound on the communication delays~\cite{PSS16}.
This series of work reveal an important limitation of Bitcoin: delaying propagation of blocks can 
waste the computational effort of correct nodes by letting them mine blocks unnecessarily 
at the same index of the chain. In this case, the attacker does not need more mining power 
than the correct miners, but simply needs to expand its local blockchain faster than the growth of the longest branch of the correct blockchain.

Ethereum proposed the {\sc Ghost} protocol to cope with this issue~\cite{SZ15}. 
The idea is simply to account for the blocks proposed by correct miners in the multiple 
branches of the correct blockchain to select the main branch. As a result, growing a branch the fastest is not sufficient for an attacker of Ethereum to be able to double spend.
Even though the propagation strategy of Ethereum differs from the pull
strategy of Bitcoin, some network attacks against Bitcoin could affect
Ethereum. 
In the Eclipse attack~\cite{HKZG15} the attacker forces the victim to connect to 8 of its malicious 
identities.
The Ethereum adaptation would require to forge $3\times$ more identities and force as many 
connections as the default number of clients is 25.
%
Apostolaki et al.~\cite{AZV16} proposed a BGP hijacking attack and showed that
the number of Internet prefixes that need to be 
hijacked for the attack to succeed depends on the distribution of the mining power. 
BGP-hijacking typically requires the control of network operators but is independent from Bitcoin 
and could potentially be exploited to delay network messages and execute a Balance attack 
in Ethereum.

The R3 consortium has been experimenting Ethereum since more than half a year now and our 
discussion with the R3 consortium indicated that they did not investigate the dependability of the 
{\sc Ghost} consensus protocol and that they also worked with Ripple, Axoni, Symbiont. Some work already evoked the danger 
of using proof-of-work techniques in a consortium context~\cite{Gra16}. 
In particular, experiments demontrated the impossibility of ordering 
even committed transactions in an Ethereum private chain without exploring the impact of the network delay~\cite{NG16}.
As a private blockchain involves typically a known and smaller number of participants 
than a public blockchain, it is also 
well-known~\cite{CDE16,Vuc16} that many 
Byzantine Fault Tolerance (BFT) solutions~\cite{CL02,PP11,LNZ16,CSV16} could be used instead.
At the time of writing, R3 has just released Corda~\cite{BCG16} as a proposed solution for 
private chains.
Corda does not yet recommend a particular consensus protocol but mentions BFT and favors modularity by allowing to plug any consensus protocol instead~\cite{BCG16}. 
Our work confirms that proof-of-work, besides being unnecessary 
for consortium private chain when the set of participants is known, is not recommended especially 
for dependability reasons.

\section{Conclusion}\label{sec:conclusion}

In this paper, we show how main proof-of-work blockchain protocols can be badly suited for consortium blockchains.
To this end, we propose the Balance attack a new attack that combines mining power with 
communication delay to affect prominent proof-of-work blockchain protocols like 
Ethereum and Bitcoin. 
This attack simply consists of convincing correct nodes 
to disregard particular proposed series of blocks to lead to a double spending.
We analyzed the tradeoff inherent to Ethereum between communication delay and mining power,
hence complementing previous observations made on Bitcoin. 

There are several ways to extend this work. First, the context is highly dependent on
medium-scale settings where statistics about all participants can be easily collected. 
It would be interesting to extend these results when the mining power of participants 
is unknown. Second, the success of the Balance attack despite a low mining power requires
communication delay between communication subgraphs. The next step is to 
compare denial-of-service and man-in-the-middle attacks and evaluate their effectiveness in introducing this delay.

\subsection*{Acknowledgements.}

We are grateful to the R3 consortium for sharing information regarding their Ethereum testbed and to Seth Gilbert and Tim Swanson for comments on an earlier version of this paper.

\bibliographystyle{IEEEtranS}
\bibliography{bib}

\begin{thebibliography}{10}
\providecommand{\url}[1]{#1}
\csname url@samestyle\endcsname
\providecommand{\newblock}{\relax}
\providecommand{\bibinfo}[2]{#2}
\providecommand{\BIBentrySTDinterwordspacing}{\spaceskip=0pt\relax}
\providecommand{\BIBentryALTinterwordstretchfactor}{4}
\providecommand{\BIBentryALTinterwordspacing}{\spaceskip=\fontdimen2\font plus
\BIBentryALTinterwordstretchfactor\fontdimen3\font minus
  \fontdimen4\font\relax}
\providecommand{\BIBforeignlanguage}[2]{{%
\expandafter\ifx\csname l@#1\endcsname\relax
\typeout{** WARNING: IEEEtranS.bst: No hyphenation pattern has been}%
\typeout{** loaded for the language `#1'. Using the pattern for}%
\typeout{** the default language instead.}%
\else
\language=\csname l@#1\endcsname
\fi
#2}}
\providecommand{\BIBdecl}{\relax}
\BIBdecl

\bibitem{AZV16}
M.~Apostolaki, A.~Zohar, and L.~Vanbever, ``Hijacking bitcoin: Large-scale
  network attacks on cryptocurrencies,'' arXiv, Tech. Rep. 1605.07524, 2016.

\bibitem{BDEWW13}
T.~Bamert, C.~Decker, L.~Elsen, R.~Wattenhofer, and S.~Welten, ``Have a snack,
  pay with bitcoins,'' in \emph{13th {IEEE} International Conference on
  Peer-to-Peer Computing, {IEEE} {P2P} 2013, Trento, Italy, September 9-11,
  2013, Proceedings}, 2013, pp. 1--5.

\bibitem{Bla02}
\BIBentryALTinterwordspacing
A.~Black, ``Hashcash - a denial of service counter-measure,'' Cypherspace,
  Tech. Rep., 2002. [Online]. Available:
  \url{http://www.hashcash.org/papers/hashcash.pdf}
\BIBentrySTDinterwordspacing

\bibitem{BCG16}
R.~G. Brown, J.~Carlyle, I.~Grigg, and M.~Hearn, ``Corda: An introduction,''
  2016.

\bibitem{Cac01}
C.~Cachin, ``Distributing trust on the internet,'' in \emph{Proceedings of the
  International Conference on Dependable Systems and Networks ({DSN})}, 2001,
  pp. 183--192.

\bibitem{CSV16}
C.~Cachin, S.~Schubert, , and M.~Vukolic, ``Non-determinism in byzantine
  fault-tolerant replication,'' in \emph{Proceedings of the 20th International
  Conference on Principles of Distributed Systems (OPODIS)}, 2016.

\bibitem{CL02}
M.~Castro and B.~Liskov, ``Practical byzantine fault tolerance and proactive
  recovery,'' \emph{ACM Trans. Comput. Syst.}, vol.~20, no.~4, pp. 398--461,
  2002.

\bibitem{CDE16}
K.~Croman, C.~Decker, I.~Eyal, A.~E. Gencer, A.~Juels, A.~Kosba, A.~Miller,
  P.~Saxena, E.~Shi, E.~G. Sirer, D.~Song, and R.~Wattenhofer, ``On scaling
  decentralized blockchains,'' in \emph{3rd Workshop on Bitcoin Research
  (BITCOIN), Barbados}, February 2016.

\bibitem{DW13}
C.~Decker and R.~Wattenhofer, ``Information propagation in the bitcoin
  network,'' in \emph{Proc. of the IEEE International Conference on
  Peer-to-Peer Computing}, 2013.

\bibitem{EGSvR16}
I.~Eyal, A.~E. Gencer, E.~G. Sirer, and R.~van Renesse, ``Bitcoin-{NG}: A
  scalable blockchain protocol,'' in \emph{13th USENIX Symposium on Networked
  Systems Design and Implementation (NSDI)}, 2016.

\bibitem{EG14}
I.~Eyal and E.~G. Sirer, ``Majority is not enough: Bitcoin mining is
  vulnerable,'' in \emph{Proceedings of the 18th Int'l Conference Financial
  Cryptography and Data Security (FC)}, 2014, pp. 436--454.

\bibitem{Fin11}
\BIBentryALTinterwordspacing
H.~Finney, ``Finney's attack,'' February 2011. [Online]. Available:
  \url{https://bitcointalk.org/index.php?topic=3441.msg48384#msg48384}
\BIBentrySTDinterwordspacing

\bibitem{FLP85}
M.~J. Fischer, N.~A. Lynch, and M.~S. Paterson, ``Impossibility of distributed
  consensus with one faulty process,'' \emph{J. ACM}, vol.~32, no.~2, pp.
  374--382, Apr. 1985.

\bibitem{GKL15}
J.~A. Garay, A.~Kiayias, and N.~Leonardos, ``The bitcoin backbone protocol:
  Analysis and applications,'' in \emph{34th Annual Int'l Conf. on the Theory
  and Applications of Crypto. Techniques}, 2015, pp. 281--310.

\bibitem{GKK16}
J.~G\"obel, H.~Keeler, A.~Krzesinski, and P.~Taylor, ``Bitcoin blockchain
  dynamics: The selfish-mine strategy in the presence of propagation delay,''
  \emph{Performance Evaluation}, Juy 2016.

\bibitem{Gra16}
V.~Gramoli, ``On the danger of private blockchains,'' in \emph{Workshop on
  Distributed Cryptocurrencies and Consensus Ledgers (DCCL'16)}, 2016.

\bibitem{HKZG15}
E.~Heilman, A.~Kendler, A.~Zohar, and S.~Goldberg, ``Eclipse attacks on
  bitcoin's peer-to-peer network,'' in \emph{24th {USENIX} Security Symposium},
  2015, pp. 129--144.

\bibitem{KAC12}
G.~Karame, E.~Androulaki, and S.~Capkun, ``Two bitcoins at the price of one?
  double-spending attacks on fast payments in bitcoin,'' \emph{{IACR}
  Cryptology ePrint Archive}, vol. 2012, p. 248, 2012.

\bibitem{LS14}
\BIBentryALTinterwordspacing
P.~Litke and J.~Stewart, ``{BGP} hijacking for cryptocurrency profit,'' August
  2014. [Online]. Available:
  \url{https://www.secureworks.com/research/bgp-hijacking-for-cryptocurrency-profit}
\BIBentrySTDinterwordspacing

\bibitem{LNZ16}
L.~Luu, V.~Narayanan, C.~Zheng, K.~Baweja, S.~Gilbert, and P.~Saxena, ``A
  secure sharding protocol for open blockchains,'' in \emph{Proceedings of the
  2016 ACM SIGSAC Conference on Computer and Communications Security}, 2016,
  pp. 17--30.

\bibitem{MLP+15}
Miller, Litton, Pachulski, Gupta, Levin, Spring, and Bhattacharjee,
  ``Discovering {B}itcoin's network topology and influential nodes,''
  University of Marylan, Tech. Rep., 2015.

\bibitem{MR95}
R.~Motwani and P.~Raghavan, \emph{Randomized Algorithms}.\hskip 1em plus 0.5em
  minus 0.4em\relax Cambridge University Press, 1995.

\bibitem{Nak08}
S.~Nakamoto, ``Bitcoin: a peer-to-peer electronic cash system,'' 2008,
  \url{http://www.bitcoin.org}.

\bibitem{NG16}
C.~Natoli and V.~Gramoli, ``The blockchain anomaly,'' in \emph{Proceedings of
  the 15th IEEE International Symposium on Network Computing and Applications
  (NCA'16)}, Oct 2016.

\bibitem{NKMS16}
K.~Nayak, S.~Kumar, A.~Miller, and E.~Shi, ``Stubborn mining: Generalizing
  selfish mining and combining with an eclipse attack,'' in \emph{{IEEE}
  European Symposium on Security and Privacy, EuroS{\&}P 2016,
  Saarbr{\"{u}}cken, Germany, March 21-24, 2016}, 2016, pp. 305--320.

\bibitem{PP11}
R.~Padilha and F.~Pedone, ``Scalable byzantine fault-tolerant storage,'' in
  \emph{{IEEE/IFIP} International Conference on Dependable Systems and Networks
  Workshops {(DSN-W} 2011), Hong Kong, China, June 27-30, 2011.}, 2011, pp.
  171--175.

\bibitem{PSS16}
R.~Pass, L.~Seeman, and A.~Shelat, ``Analysis of the blockchain protocol in
  asynchronous networks,'' Crytology ePrint Archive, Tech. Rep. 454, 2016.

\bibitem{Ros12}
M.~Rosenfeld, ``Analysis of hashrate-based double-spending,'' 2012.

\bibitem{SZ15}
Y.~Sompolinsky and A.~Zohar, ``Secure high-rate transaction processing in
  bitcoin,'' in \emph{Financial Cryptography and Data Security - 19th
  International Conference, {FC} 2015, San Juan, Puerto Rico, January 26-30,
  2015, Revised Selected Papers}, 2015, pp. 507--527.

\bibitem{Vec11}
\BIBentryALTinterwordspacing
vector76, ``The vector76 attack,'' August 2011. [Online]. Available:
  \url{https://bitcointalk.org/index.php?topic=36788.msg463391#msg463391}
\BIBentrySTDinterwordspacing

\bibitem{Vuc16}
M.~Vukol\'{i}c, ``The quest for scalable blockchain fabric: Proof-of-work vs.
  {BFT} replication,'' in \emph{Proceedings of the IFIP WG 11.4 Workshop on
  Open Research Problems in Network Security (iNetSec 2015)}, 2015, pp.
  112--125.

\bibitem{Woo15}
G.~Wood, ``Ethereum: A secure decentralised generalised transaction ledger,''
  2015, yellow paper.

\end{thebibliography}

\appendix

\subsection{Analysis for the general case}\label{app:proof}

As we show below, if the attacker can delay communications for long enough between $k$ subgraphs of the communication graph where $k \geq 2$ then the attacker does not need a large portion $\rho$ of the mining power for the blockchain system to be block oblivious. 
We consider an attacker that can delay messages of the communication graph $G$ between $k$ subgraphs $G_1, ..., G_k$, each graph owning the same portion $\frac{1-\rho}{k}$ of the total mining power.

Similarly to Section~\ref{sec:proof},
we consider that for $0 < i \leq k$, each of the 
subgraph $G_i$ mine blocks during delay $\tau$.
During that time, each $G_i$ performs a series of $n = \frac{1-\rho}{k}t\tau$ independent and identically distributed Bernoulli trials that returns one in case of success with probability 
$p = \frac{1}{d}$ and 0 otherwise. 
Let the sum of these outcomes be the random variable $X_i$ ($0<i\leq k$) with a binomial distribution and mean: 
\begin{equation}
\mu_c = np = \frac{(1-\rho)t\tau}{kd}. \label{eq:mean}
\end{equation}

Similarly, the mean of the number of blocks mined by the malicious miner during time $\tau$ is
\begin{equation}
\mu_m = \frac{\rho t\tau}{d}. \label{eq:meanm}
\end{equation}

We are interested in measuring the probability that an attacker can select a candidate blockchain among 
the existing ones and make it adopted by the system. Note that this is easier than rewriting the 
history with the attacker personal blocks but is sufficient to double spend by simply making sure the initial transaction that spend the coins is only part of blockchains candidate that will be discarded.

The attack relies on minimizing the difference in mined blocks between any pair of communication subgraphs. 
Hence, let us denote $\Delta$ the difference of the number of blocks mined on any two subgraphs:
$$\Delta = \max_{\forall 0<i,j\leq k}(|X_i - X_j|).$$

To this end, we bound the difference between the number of blocks mined in two subgraphs so that  
the attacker can mine more than the difference. 

\begin{fact}[Bernoulli's inequality]
$1+nt \leq \left(1+t\right)^n$ for $n\geq 1$ and $t\geq -1$. 
\end{fact}

\begin{theorem}\label{thm:highproba}
An attacker decomposing the communication graph into $k$ subgraphs upper-bounds their difference 
$\Delta$ in mined blocks by 
$2\delta \mu_c$
with probability $\Pr[\Delta < 2\delta \mu_c] > 1-2ke^{-\frac{\delta^2}{3}\mu_c}$ where $\delta>0$.
\end{theorem}
\begin{proof}
This difference is less than $2\delta \mu_c$ with probability larger than the probability that all random variables are within a $\pm \delta$ multiplying factor from the mean, we have:
\begin{equation}
\Pr[\Delta < 2\delta \mu_c] \geq \Pi^{k}_{i=1} \Pr\left[|X_i - \mu_c| < \delta \mu_c\right].\label{eq:proba}
\end{equation}
The probability that the numbers of blocks mined by each subgraph are within a $\pm \delta$ factor from their mean, is bound for $0 < \delta < 1$ and $0 < i \leq k$ by Chernoff bounds~\cite{MR95}: 
\begin{equation}
\left\{\begin{array}{ll}
\Pr\left[X_i \geq (1+\delta)\mu_c\right] &\leq e^{-\frac{\delta^2}{3}\mu_c},\notag
\\
\Pr\left[X_i \leq (1-\delta)\mu_c\right]& \leq e^{-\frac{\delta^2}{2}\mu_c}.
\end{array}\right.
\end{equation}
Thus, we have 
\begin{eqnarray}
\Pr[|X_i -\mu_c| \geq \delta \mu_c] &\leq& 2e^{-\frac{\delta^2}{3}\mu_c}, \notag \\
\Pr[|X_i -\mu_c| < \delta \mu_c] &>& 1-2e^{-\frac{\delta^2}{3}\mu_c}. \notag
\end{eqnarray}
and with Eq.~\ref{eq:proba} we obtain: 

\begin{equation}
\Pr[\Delta < 2\delta \mu_c] > \left(1-2e^{-\frac{\delta^2}{3}\mu_c}\right)^k.\label{eq:power}
\end{equation}

As $\mu_c \geq 0$, we have that:
\begin{eqnarray}
{-\frac{\delta^2}{3}\mu_c} &\leq& 0 \notag \\
-e^{\frac{-\delta^2}{3}\mu_c} &\geq& -1 \notag
\end{eqnarray}
and as $k \geq 1$ we can apply the Bernoulli inequality to Eq.~\ref{eq:power}.

Hence, Eq.~\ref{eq:power} becomes: 
\begin{equation}
\Pr[\Delta < 2\delta \mu_c] > 1-2ke^{-\frac{\delta^2}{3}\mu_c}. \label{eq:result}
\end{equation}
\end{proof}

The next theorem bounds the number of blocks the attacker needs to mine to force the system to adopt 
the candidate blockchain of its choice.
\begin{theorem}
If the attacker delays the links of $E_0$ while the miners on each of the $k$ subgraphs mine for
$$\tau \geq \frac{3kd\log(\frac{2k}{\varepsilon})}{\delta^2(1-\rho)t}$$ seconds,
then with probability $1-\varepsilon$ the difference between the number of blocks mined on the two subgraphs is lower than $\frac{2\delta (1-\rho)t\tau}{kd}$. 
\end{theorem}
\begin{proof}
The proof relies on upper bounding 
$2ke^{-\frac{\delta^2}{3}\mu_c}$
by $\varepsilon$:
\begin{eqnarray}
2ke^{-\frac{\delta^2}{3}\mu_c} &\leq& \varepsilon, \notag \\
\mu_c &\geq& \frac{3\log(\frac{2k}{\varepsilon})}{\delta^2}. \notag
\end{eqnarray}
By replacing $\mu_c$ by the expression of Eq.~\ref{eq:mean}, we obtain:
\begin{eqnarray}
\tau &\geq& \frac{3kd\log(\frac{2k}{\varepsilon})}{\delta^2(1-\rho)t}. \notag
\end{eqnarray}
\end{proof}

Let us now bound the probability that the number of blocks mined on $k$ subgraphs is always lower than the expectation of the number of blocks mined by the malicious node.
\begin{theorem}
If $\delta = \frac{\mu_m -1}{2\mu_c}$ then $$\Pr[\Delta < \mu_m] > 1-2ke^{-\frac{\left(\mu_m-1\right)^2}{12\mu_c}}.$$
\end{theorem}
\begin{proof}
The proof follows from replacing $\delta$ by $\frac{\mu_m -1}{2\mu_c}$ in Eq.~\ref{eq:result}.
\end{proof}

\end{document}